\documentclass[a4paper,10pt]{article}
\usepackage{amssymb,amsmath,amsthm}
\usepackage{fullpage}
\usepackage{hyperref}
\usepackage{eucal}
\usepackage{color}
\usepackage{stackrel}
\usepackage{graphicx}

\newcommand{\Real}{\mathbb{R}}
\newcommand{\C}{\mathbb{C}}

\newcommand{\supp}{\mathop{\mathrm{supp}}\nolimits}
\newcommand{\Dom}{\mathsf{D}}

\newcommand{\sii}{L^2}

\newcommand{\der}{\mathrm{d}}

\newenvironment{psmallmatrix}
  {\left(\begin{smallmatrix}}
  {\end{smallmatrix}\right)}
\begingroup
    \makeatletter
    \@for\theoremstyle:=definition,remark,plain\do{%
        \expandafter\g@addto@macro\csname th@\theoremstyle\endcsname{%
            \addtolength\thm@preskip\parskip
            }%
        }
\endgroup
\newtheorem{Theorem}{Theorem}
\newtheorem{Lemma}{Lemma}
\newtheorem{Proposition}{Proposition}
\newtheorem{Corollary}{Corollary}

\theoremstyle{definition}
\newtheorem{Remark}{Remark}

\def\OMIT#1{}
\usepackage[normalem]{ulem}
\definecolor{DarkGreen}{rgb}{0,0.5,0.1} 

\newcommand\soutD{\bgroup\markoverwith
{\textcolor{DarkGreen}{\rule[.5ex]{2pt}{1pt}}}\ULon}
\newcommand\soutP{\bgroup\markoverwith
{\textcolor{blue}{\rule[.5ex]{2pt}{1pt}}}\ULon}
\newcommand{\Hm}[1]{\leavevmode{\marginpar{\tiny%
$\hbox to 0mm{\hspace*{-0.5mm}$\leftarrow$\hss}%
\vcenter{\vrule depth 0.1mm height 0.1mm width \the\marginparwidth}%
\hbox to
0mm{\hss$\rightarrow$\hspace*{-0.5mm}}$\\\relax\raggedright #1}}}

\begin{document}
%
\title{\textbf{\LARGE
Spectral properties of two-dimensional half-space semi-Dirac semi-metals
}}
\author{Tuyen Vu }
\date{\small 
%
\begin{quote}
\centering
\emph{
Department of Mathematics, Faculty of Nuclear Sciences and 
Physical Engineering, Czech Technical University in Prague, 
Trojanova 13, 12000 Prague 2, Czechia;\\
thibichtuyen.vu@fjfi.cvut.cz.%
}
\end{quote}
27 February 2024
}
\maketitle
\vspace{-5ex} 
\begin{abstract}
\noindent
The paper deals with the semi-Dirac operator in a half-space arising  in the description of quasi-particles in quantum mechanics as well as in semi-metals materials and related structures.
It completely shows the  self-adjointness, computes the square and the spectrum of the operator. 
We also set up sufficient conditions for the existence of the point and discrete spectrum  by including some appropriate potentials in the problem and study the spectral stability properties for perturbed operators. 
 
%
%
\end{abstract}


%

\section{Introduction}
Semi-Dirac semi-metals are materials studied in the field of condensed matter physics  with nanostructures can be featured as a combination of linear   and 
quadratic dispersions and exhibiting  tunable electronic, optical and transport properties. In spite of garnering substantial research attention, e.g. \cite{Ba,KA,YS}, the research of semi-Dirac semi-metals  still remains a recent, fascinating and challenging topic in quantum materials.

Given a half-space domain ~$\Omega = \Real \times (0,\infty) := \Real^2_+$
see Figure~\ref{halfspace},
on which the semi-Dirac semi-metals  operator acts as 
\begin{equation}
  T_{\Real^2_+} := 
  \begin{pmatrix}
     -i \partial_y & -\partial^2_x + \delta \\
     -\partial^2_x + \delta & i \partial_y
  \end{pmatrix}
  \qquad \mbox{in} \qquad
  \sii(\Real^2_+;\C^2)
  \,,
\end{equation}
here some physical constants are assumed to be equal to $1$ by virtue of re-scaling the system and  the operator domain is defined as follows,
$$
  \Dom(T_{\Real^2_+}) = \left\{ 
  u \in H^{1}(\Real^2_+;\C^2) : \ \partial_x^2 u \in \sii(\Real^2_+;\C^2) \mbox{ and }
  u_1 = u_2 \ \mbox{on } y=0 
  \right\}
.$$
\begin{figure}[h]
  \begin{center}
  \includegraphics[scale=1.5]{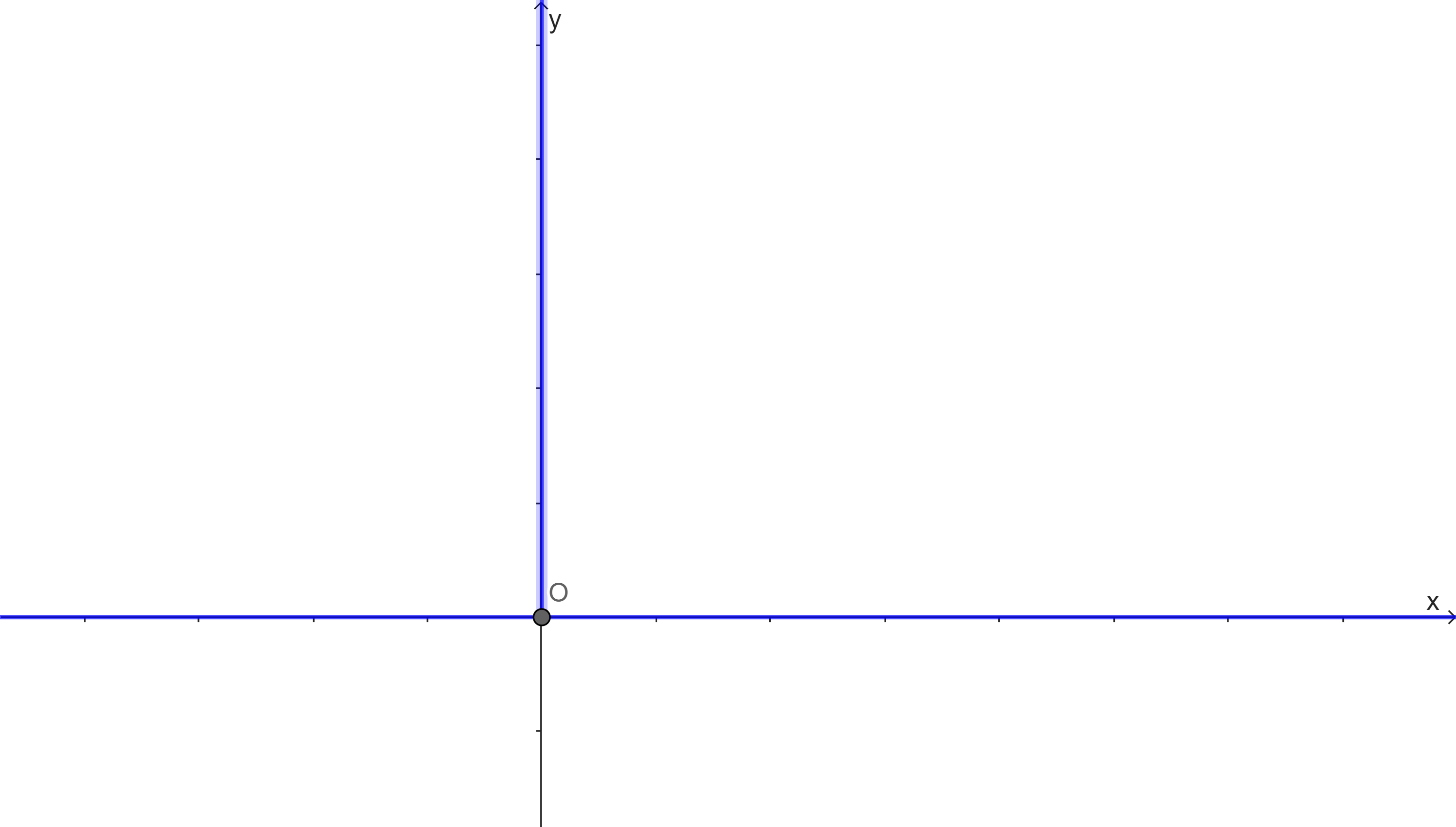}
  \caption{The half-space $\Real^2_+$.
  }
  \label{halfspace}
  \end{center}
\end{figure}

If the domain $\Omega = \Real^2$ then the corresponding operator  becomes 
\begin{equation}\label{operator on whole space}
  T_{\Real^2} := 
  \begin{pmatrix}
     -i \partial_y & -\partial^2_x + \delta \\
     -\partial^2_x + \delta & i \partial_y
  \end{pmatrix}
  \qquad \mbox{in} \qquad
  \sii(\Real^2;\C^2)
  \,,
\end{equation}
 \begin{equation}
  \Dom(T_{\Real^2}) :=  
  \Big\{
  \psi = 
  \begin{psmallmatrix}
  \psi_1 \\ \psi_2 
  \end{psmallmatrix}
  \in H^{1}(\Real^2;\C^2) : \ \partial_x^2\psi_2, \partial_x^2\psi_1 \in \sii(\Real^2) 
  \Big\} \,.
\end{equation}
As studied in \cite{KA}, $T_{\Real^2}$ is self-adjoint and its spectrum is purely absolutely continuous $$\sigma(T_{\Real^2}) = (-\infty, -\delta] \cup [\delta, \infty).$$
Being motivated by \cite{KA}, we study the self-adjointness, spectral properties of $T_{\Real^2_+}$ 
and illustrate the existence of the discrete spectrum as well as the effect of some potentials on the spectrum of the semi-Dirac semi-metals problem. 

The paper is organised as follows.  
In Section~\ref{Sec.no}, we demonstrate the formula for the expectation value of the square of the operator. The self-adjointness of the semi-Dirac semi-metals operator in a half-space
is established in Section~\ref{Sec.self-adjointness}. 
The results providing insights into  the spectral properties of the operator, especially for the existence of the discrete spectrum  and the spectral instability properties for perturbed operators are illustrated in Section \ref{Sec.proof}.

\section{The square of the semi-Dirac semi-metals operator in $\Real^2_+$}\label{Sec.no}
 It is remarkable that $\Dom(T_{\Real^2_+})$ is a Banach space endowed with the norm 
$$ \|.\|_{T_{\Real^2_+; \C^2}} := \sqrt{\|.\|^2_{H^1({\Real^2_+; \C^2})} + \|\partial_x^2 .\|^2_{\sii({\Real^2_+; \C^2})}} \, .$$

In addition, denote
$$ \|.\|_{T_{\Real^2_+}} := \sqrt{\|.\|^2_{H^1({\Real^2_+; \C})} + \|\partial_x^2 .\|^2_{\sii({\Real^2_+; \C})}} \, , $$
 $$\mathtt{H}^1(\Real^2_+) := \overline{C^\infty(\overline{\Real^2_+}; \C)}^{\|.\|_{T_{\Real^2_+}}} \,, \quad \quad \quad \quad \, \,  \mathtt{H}^1_0(\Real^2_+) := \overline{C_0^\infty(\Real^2_+; \C)}^{\|.\|_{T_{\Real^2_+}}} \,,$$
$$\mathtt{H}^1(\Real^2_+; \C^2) := \{ u = \begin{psmallmatrix} u_1 \\ u_2 \end{psmallmatrix}, u_1, u_2 \in \mathtt{H}^1(\Real^2_+)\} \,, \quad \quad
\mathtt{H}_0^1(\Real^2_+; \C^2) := \{ u = \begin{psmallmatrix} u_1 \\ u_2 \end{psmallmatrix}, u_1, u_2 \in \mathtt{H}_0^1(\Real^2_+)\} \,.$$
We have $$\mathtt{H}^1(\Real^2_+) = 
\{ u \in H^1(\Real^2_+) : \partial_x^2 u \in L^2(\Real^2_+) \}.$$
The proof for this result can be analogously found  in \cite{Adam,Grisvard}. To study the space $\mathtt{H}^1_0(\Real^2_+)$ thoroughly, we apply the approach and cut-off functions  defined in \cite{Evans}.
\begin{Lemma}\label{density}
\begin{equation}
\mathtt{H}^1_0(\Real^2_+) = 
\{ u \in H^1_0(\Real^2_+) : \partial_x^2 u \in L^2(\Real^2_+) \} \,.
\end{equation}
\end{Lemma}
\begin{proof}
First of all, we easily check that 
$$\mathtt{H}^1_0(\Real^2_+) \subset
\{ u \in H^1_0(\Real^2_+) : \partial_x^2 u \in L^2(\Real^2_+) \} \,.$$

For the inverse inclusion, take a function $u \in H^1_0(\Real^2_+)$ such that $ \partial_x^2 u \in L^2(\Real^2_+)$ then there exists a sequence $\{u_n\}_{n\geq 1} \subset C^\infty(\overline{\Real^2_+})$: $ u_n \xrightarrow[n \to \infty]{} u$ in $\mathtt{H}^1(\Real^2_+)$. By virtue of the trace theorem, we have $u_n(x,0) \xrightarrow[n \to \infty]{} u(x,0)$ in $H^\frac{1}{2}(\Real)$.

For each $x \in \Real, y_k \in [0, \infty)$, we estimate
\begin{equation}
|u_n(x,y_k)| \leq |u_n(x,0)| + \int_0^{y_k} |\partial_y u_n(x,t)| \der t \,.
\end{equation}
Hence, 
\begin{equation}
\int_\Real |u_n(x,y_k)|^2 \der x \leq 2\Big( \int_\Real |u_n(x,0)|^2
\der x + y_k \int_0^{y_k} \int_\Real |\partial_y u_n(x,t)|^2 \der x \der t \Big) \,.
\end{equation}
Letting $ n \rightarrow \infty$, we have
\begin{equation}\label{bdt}
\int_\Real |u(x,y_k)|^2 \der x \leq 2 y_k \int_0^{y_k} \int_\Real |\partial_y u(x,t)|^2 \der x \der t 
\end{equation}
for a.e. $y_k > 0$.

Take $  \varphi(y) \in C^\infty[0, \infty) , 0\leq \varphi \leq 1$
$$ \varphi(y) : =  \left\{\begin{array}{rcl}
	&1 \quad \quad &\mbox{ if } y \in [0, 1]  \,,\\
    &0  \quad\quad &\mbox{ if } y \geq 2  \,.\\
\end{array}\right.$$ 
Putting $$ \varphi_n(y) : = \varphi(n y), \quad \quad \quad v_n := u(x)(1-\varphi_n) \,.$$
Thus, 
$$ \partial_x v_n = \partial_x u (1-\varphi_n), \quad \quad \quad \partial_x^2 v_n = \partial_x^2 u (1-\varphi_n) $$
and $$ \partial_y v_n = \partial_y u (1-\varphi_n) - u (\varphi_n)' \,.$$
As a result, 
$$ \int_{\Real^2_+} |\partial_x v_n - \partial_x u|^2 \leq \int_{\Real^2_+} |\varphi_n|^2 |\partial_x u|^2 \der x \der y \xrightarrow[n \to \infty]{} 0,$$
 $$  \int_{\Real^2_+} |\partial_x^2 v_n - \partial_x^2 u|^2 \leq \int_{\Real^2_+} |\varphi_n|^2 |\partial_x^2 u|^2 \der x \der y \xrightarrow[n \to \infty]{} 0$$
due to the dominated convergence theorem.
Moreover, the support of $\varphi_n$, $\supp\varphi_n \subset \Real \times [0, \frac{2}{n}]$ and taking \eqref{bdt} into account, we have
\begin{equation}
\begin{aligned}
\int_{\Real^2_+} |\partial_y v_n - \partial_y u|^2 \der x \der y &\leq \int_{\Real^2_+} |\varphi_n|^2 |\partial_y u|^2  \der x \der y + \int_0^\infty \int_\Real  |u|^2 |\varphi_n'|^2 \der x \der y \,,\\
&= \int_{\Real^2_+} |\varphi_n|^2 |\partial_y u|^2  \der x \der y + \int_0^\frac{2}{n} \int_\Real n^2 |u|^2 |\varphi'|^2 \der x \der y \,,\\
&\leq \int_{\Real^2_+} |\varphi_n|^2 |\partial_y u|^2  \der x \der y + 2 n^2 \sup_{[0,\infty]}|\varphi'|^2 \Big( \int_0^\frac{2}{n} y \der y \Big) \Big( \int_0^\frac{2}{n} \int_\Real |\partial_y u (x,y)|^2 \der x \der y \Big) \,,\\
&\xrightarrow[n \to \infty]{} 0 \,.
\end{aligned}
\end{equation} 

Therefore, $v_n \xrightarrow[n \to \infty]{} u$ in $\mathtt{H}^1(\Real^2_+)$. In addition, $ v_n = 0$ if $0 \leq y \leq \frac{1}{n}$. Thus, we can modify the functions $u_n \in C^\infty_0(\Real^2_+)$ such that $ u_n \xrightarrow[n \to \infty]{} u$ in $\mathtt{H}^1(\Real^2_+)$. It deduces that $u \in \mathtt{H}^1_0(\Real^2_+)$. It concludes the proof.

\end{proof}

Now we would like to to establish the square of the operator as analogously  described in \cite{Tuyen}. To write computations more convenient, we denote $\|.\|$ stands for the $\sii(\Real^2_+)$-norm.
\begin{Theorem}\label{formula}
For every $u = \begin{psmallmatrix} u_1 \\ u_2 \end{psmallmatrix} \in \Dom(T_{\Real^2_+})$,
\begin{equation}\label{norm}   
  \|T_{\Real^2_+} u\|^2   = \|\partial_y u\|^2 + \|\partial_x^2 u\|^2 + 2 \delta\|\partial_x u\|^2 
  + \delta^2 \|u\|^2.
\end{equation}
\end{Theorem}
\begin{proof}  First of all, we will show that \begin{equation}\label{dense} C^\infty(\Real^2) \cap \Dom(T_{\Real^2_+}) \mbox{ is a core of } T_{\Real^2_+}.
\end{equation}
Recall that $\mathtt{H}^1(\Real^2_+) := \overline{C^\infty(\overline{\Real^2_+})}^{\|.\|_{T_{\Real^2_+}}}$
 then for each $u \in \Dom(T_{\Real^2_+})$
there exists a sequence $\{u_{1n}\}_{n \geq 1} \subset C^\infty(\overline{\Real^2_+})$ such that $ u_{1n} \xrightarrow[n \to \infty]{} u_1$ in $\mathtt{H}^1(\Real^2_+)$. 

Putting $v := u_2 - u_1$ then it follows that $v \in \mathtt{H}^1_0(\Real^2_+)$. Taking Lemma \ref{density} into account, there exists a sequence $\{v_n\}_{n\geq 1} \in C^\infty_0(\Real^2_+)$
such that $ v_{n} \xrightarrow[n \to \infty]{} v$ in $\mathtt{H}^1(\Real^2_+)$. Hence, we deduce that $ u_{1n} + v_{n} \xrightarrow[n \to \infty]{} u_{2}$ in $\mathtt{H}^1(\Real^2_+)$. As a result, the smooth sequence $\{ w_n : = \begin{psmallmatrix} u_{1n} \\ u_{1n} + v_n \end{psmallmatrix} \}_{n\geq 1} \xrightarrow[n \to \infty]{} u$ in $\Dom(T_{\Real^2_+})$. As a result, $ T_{\Real^2_+}w_n \xrightarrow[n \to \infty]{} T_{\Real^2_+}u$ in $L^2(\Real^2_+)$. Therefore, \eqref{dense} holds and thus we can compute the square of the operator acting on the smooth functions without loss of generality. 
 We have

\begin{equation*}
\begin{aligned}
\|T_{\Real^2_+} u \|^2 
&= \|-i \partial_y u_{1} + (-\partial^2_x + \delta) u_{2}\|^2 + \|(-\partial^2_x + \delta)u_{1} + i \partial_y u_{2}\|^2 \,,\\
&= \|-i \partial_y u_{1}\|^2 + \|(-\partial^2_x + \delta) u_{2}\|^2 + 2\Re (-i \partial_y u_1, (-\partial^2_x + \delta) u_2)\\
& \quad + \|(-\partial^2_x + \delta)u_{1}\|^2 + \|i \partial_y u_{2}\|^2 
+ 2 \Re ((-\partial^2_x + \delta)u_{1}, i \partial_y u_{2}) \,,\\
&= \|\partial_y u\|^2 + \delta^2 \|u\|^2 + \|\partial_x^2 u\|^2 + 2\delta \Re (-\partial_x^2 u_{1}, u_{1}) + 2\delta \Re (-\partial_x^2 u_{2}, u_{2}) \\
& \quad + 2\Re (-i \partial_y u_{1}, (-\partial^2_x + \delta) u_{2}) + 2 \Re ((-\partial^2_x + \delta)u_{1}, i \partial_y u_{2})\,,\\
&= \|\partial_y u\|^2 + \delta^2 \|u\|^2 + \|\partial_x^2 u\|^2 + 2\delta \|\partial_x u\|^2 \\
& \quad + 2\Re (-i \partial_y u_{1}, (-\partial^2_x + \delta) u_{2}) + 2 \Re ((-\partial^2_x + \delta)u_{1}, i \partial_y u_{2})\,.
\end{aligned}
\end{equation*}
Regarding the boundary conditions and integration by parts, we estimate
\begin{equation*}
\begin{aligned}
\Re (-i \partial_y u_{1}, (-\partial^2_x + \delta) u_{2}) &= -\Re\int_{\Real^2_+} i \partial_y \overline{u_{1}} \partial^2_x u_{2} \ \der x \der y + \delta \Re \int_{\Real^2_+} i \partial_y \overline{u_{1}} u_{2} \ \der x \der y \,, \\
&= \Re \int_{\Real^2_+} i \overline{u_1} \partial_y\partial_x^2 u_2 \der x \der y + \Re i \int_\Real (\overline{u_1} \partial_x^2 u_2)|_{y=0} \der x \\
& \quad - \delta \Re \int_{\Real^2_+} i  \overline{u_{1}} \partial_y u_{2} \ \der x \der y  - \delta \Re \int_\Real i (\overline{u_1} u_2)|_{y=0} \der x    \,, \\
&= \Re \int_{\Real^2_+} i \overline{u_1} \partial_y\partial_x^2 u_2 \der x \der y - \Re i \int_\Real |\partial_x u_2|^2|_{y=0} \der x \\
&\quad - \delta \Re \int_{\Real^2_+} i  \overline{u_{1}} \partial_y u_{2} \ \der x \der y + \delta \Re i \int_\Real |u_1|_{y=0} \der x \,,\\
&= \Re \int_{\Real^2_+} i \partial_x^2\overline{u_1} \partial_y u_2 \der x \der y - \delta \Re \int_{\Real^2_+} i  \overline{u_{1}} \partial_y u_{2} \ \der x \der y \,,\\
&= - \Re ((-\partial^2_x + \delta)u_{1}, i \partial_y u_{2})\,.
\end{aligned}
\end{equation*}

Combining all previous computations, we get the desired formula. This completes the proof.

\end{proof}
\begin{Remark}
For other domains, e.g. rectangular domains, the density \eqref{dense} may not hold due to the fact that the validity of Lemma \ref{density} is not true for general sets.
\end{Remark}

\section{The self-adjointness of the semi-Dirac semi-metals operator}\label{Sec.self-adjointness}
First of all, we study on the closeness of $T_{\Real^2_+}$. The result can be directly gained by using the definition of the operator closedness. 
\begin{Proposition}
$T_{\Real^2_+}$ is a closed operator.
\end{Proposition}
\begin{proof}
Take a sequence $\{u_n\}_{n \geq 1} \in \Dom(T_{\Real^2_+}) :u_n \xrightarrow[n \to \infty]{} u $ and $T_{\Real^2_+}u_n \xrightarrow[n \to \infty]{} v$. Then we have $u_n \xrightarrow[n \to \infty]{} u \in H^1(\Real^2_+)$ and there exists $w \in L^2(\Real^2_+)$ such that $\partial_x^2 u_n \xrightarrow[n \to \infty]{} w$ in $L^2(\Real^2_+)$. In addition, for all $\varphi \in C^\infty_0(\Real^2_+)$
\begin{equation}
\left\langle u, \partial_x^2\varphi\right\rangle = \lim_{n\rightarrow\infty}\left\langle u_n, \partial_x^2\varphi\right\rangle = \lim_{n\rightarrow\infty}\left\langle\partial_x^2 u_n, \varphi\right\rangle = \left\langle w, \varphi\right\rangle \,,
\end{equation}
where $\langle\cdot,\cdot\rangle $ denotes the duality bracket of distributions.
It deduces that $\partial^2_x u = w\in L^2(\Real^2_+)$. Therefore, $u \in \Dom(T_{\Real^2_+})$ and $T_{\Real^2_+}u_n \xrightarrow[n \to \infty]{} T_{\Real^2_+}u$. Thus, $T_{\Real^2_+}$ is closed. The proof is completed.
\end{proof}

Recall that $T_{\Real^2}$ is self-adjoint. To point out the self-adjointness of $T_{\Real^2_+}$, we employ the extension map as the approach described in \cite{Arrizibalaga-LeTreust-Raymond_2018}. We will establish the trace continuous maps  

Putting
\begin{equation}
  T := 
  \begin{pmatrix}
     -i \partial_y & -\partial^2_x + \delta \\
     -\partial^2_x + \delta & i \partial_y
  \end{pmatrix}
  \,,
\end{equation}
acts as a differential operator with  distributional sense.

\begin{Lemma}\label{denseM}
Let $\mathcal{M} := \{ u \in L^2(\Real^2_+; \C^2)| Tu \in L^2(\Real^2_+; \C^2)\}$ endowed with the norm 
$$\|u\|_\mathcal{M}^2 :=  \|u\|_{L^2(\Real^2_+)}^2 + \|Tu\|_{L^2(\Real^2_+)}^2$$
 then $\mathcal{M}$ is a Hilbert space and $C^\infty(\overline{\Real^2_+}; \C^2)$ is dense in $\mathcal{M}$.
\end{Lemma}
\begin{proof}
The proof is similar as that for the Dirac operators, see  \cite{Benguria-Fournais-Stockmeyer-Bosch_2017b}. Indeed, it is obvious that $\mathcal{M}$ is a Hilbert space encoded with a scalar product 
$$(u,v)_{\mathcal{M}} = (u,v)_{L^2(\Real^2_+)} + (Tu, Tv)_{L^2(\Real^2_+)} \,.$$ 
To prove $C^\infty(\overline{\Real^2_+}; \C^2)$ is dense in $\mathcal{M}$, we take $f \in \mathcal{M}$ such that $(f,g)_{\mathcal{M}}=0$ for all $g \in C^\infty(\overline{\Real^2_+}; \C^2)$. Putting $h := Tf$ then $Th = -f$ in the sense of distributions.
Denote by $\tilde{f}, \tilde{h}$  the trivial extensions of $f,h$ to $\Real^2$. For all $\phi \in C^\infty_0(\Real^2; \C^2)$ we have
$$ <T\tilde{h}, \phi> =   (\tilde{h},T\phi)_{L^2(\Real^2)} = (h,T\phi)_{L^2(\Real^2_+)}= (-f,\phi)_{L^2(\Real^2_+)}=(-\tilde{f},\phi)_{L^2(\Real^2)}.$$
It yields that $T\tilde{h} = -\tilde{f} \in L^2(\Real^2; \C^2)$. Therefore, $\tilde{h} \in H^1(\Real^2; \C^2), \partial_x^2 \tilde{h} \in L^2(\Real^2)$ by means of the Fourier transform.
As a result, $h \in H^1(\Real^2_+; \C^2), \partial^2_x h \in L^2(\Real^2_+; \C^2)$ and $h=0$ on $y=0$ due to the trace theorem. Equivalently, $h \in \mathtt{H}^1_0(\Real^2_+; \C^2)$ then there exists a sequence $h_n \in C^\infty_0(\Real^2_+; \C^2) : h_n \xrightarrow[n \to \infty]{} h$ in $\mathtt{H}_0^1(\Real^2_+; \C^2)$.

For all $u \in \mathcal{M}$, we have
\begin{equation}
\begin{aligned}
(f,u)_{\mathcal{M}} &= (f,u)_{L^2(\Real^2_+)} + (h, Tu)_{L^2(\Real^2_+)} 
= (f,u)_{L^2(\Real^2_+)} + \lim_{n\rightarrow \infty} (h_n, Tu)_{L^2(\Real^2_+)} \,,\\
&= (f,u)_{L^2(\Real^2_+)} + \lim_{n\rightarrow \infty} \overline{<Tu, h_n>} \,,\\
&= (f,u)_{L^2(\Real^2_+)} + \lim_{n\rightarrow \infty} (Th_n, u)_{L^2(\Real^2_+)} 
= (f,u)_{L^2(\Real^2_+)} +  (Th, u)_{L^2(\Real^2_+)} \,,\\
&= 0 \,.
\end{aligned}
\end{equation}
Hence, $f =0$ and thus, $C^\infty(\overline{\Real^2_+}; \C^2)$ is dense in $\mathcal{M}$. The proof is completed.

\end{proof}

In order to establish a continuous trace map for $\mathcal{M}$, we recall the definition of Fourier transform of a function $f(t) \in L^2(\Real^N)$
\begin{equation*}
\hat{f}(\xi) = \int_{\Real^N} e^{-ix\xi} f(t) \der t \,,
\end{equation*}
and its inverse transform
\begin{equation*}
f(t) = \frac{1}{(2\pi)^N}\int_{\Real^N} e^{ix\xi} \hat{f}(\xi) \der \xi \,.
\end{equation*}
We have the Plancherel identity, see \cite{JN}
$$\int_{\Real^N} |\hat{f}(\xi)|^2 \der \xi = (2\pi)^N\int_{\Real^N} |f(x)|^2 \der x \,.$$

Putting $$\mathtt{H}^{\frac{1}{2}} := \{ u(x) \in L^2(\Real; \C) : \int_{\Real}|\hat{u}(\xi_1)|^2(1 + |\xi_1|^2 + |\xi_1|^4)^{\frac{1}{2}} \der \xi_1 := \|u\|^2_{\mathtt{H}^{\frac{1}{2}}} < \infty \}.$$
Then $\mathtt{H}^{\frac{1}{2}}$ is also a Hilbert space and its dual space is 
$$\mathtt{H}^{-\frac{1}{2}} := \{ u(x) \in L^2(\Real; \C) : \int_{\Real}|\hat{u}(\xi_1)|^2(1 + |\xi_1|^2 + |\xi_1|^4)^{-\frac{1}{2}} \der \xi_1 := \|u\|^2_{\mathtt{H}^{-\frac{1}{2}}} < \infty \}.$$
Denote 
$$\mathtt{H}^{\frac{1}{2}}_{\C^2} := \{ u(x) = \begin{psmallmatrix} u_1 \\ u_2 \end{psmallmatrix}, u_1, u_2 \in \mathtt{H}^{\frac{1}{2}} \} \,,$$
$$\mathtt{H}^{-\frac{1}{2}}_{\C^2} := \{ u(x) = \begin{psmallmatrix} u_1 \\ u_2 \end{psmallmatrix}, u_1, u_2 \in \mathtt{H}^{-\frac{1}{2}} \} \,,$$
we have the following theorem

\begin{Theorem}\label{trace0}
For every $u \in \mathtt{H}^1(\Real^2_+)$, we have $u(x,0) \in \mathtt{H}^{\frac{1}{2}}$ and $\|u(x,0)\|_{\mathtt{H}^{\frac{1}{2}}} \leq C \|u\|_{\mathtt{H}^1(\Real^2_+)}$ with $C$ is a constant. Reversely,
for each $g \in \mathtt{H}^{\frac{1}{2}}$ there exists a function $u \in \mathtt{H}^1(\Real^2_+)$ and a constant $M > 0$ such that $ g= u(x,0)$ and $ \|u\|_{\mathtt{H}^1(\Real^2_+)}\leq M  \|u(x,0)\|_{\mathtt{H}^{\frac{1}{2}}}$.
\end{Theorem}
\begin{proof}
For the first statement, we use similar arguments for Sobolev spaces as described in \cite{JN} to prove the theorem, we can suppose $u \in C^\infty(\overline{\Real^2_+})$ without loss of generality due to the density of smooth functions up to the boundary in $\mathtt{H^1}(\Real^2_+)$. First of all, we put $g(x) := u(x,0)$, extend $u$ to $\Real^2$
\begin{equation}
\tilde{u}(x,y) = \left\{\begin{array}{rcl}
	u(x,y) &\mbox{ if } y \geq 0  \,,\\
	 u(x,-y) &\mbox{ if } y < 0  \,,\\
\end{array}\right.
\end{equation}
then $\tilde{u} \in H^1(\Real^2)$ and $\partial_x^2 \tilde{u} \in L^2(\Real^2)$. Denote $\hat{u}$ is the Fourier transform of $\tilde{u}$, we deduce that
$$\hat{g}(\xi_1) = \frac{1}{2\pi} \int_{\Real} \hat{u}(\xi) \der \xi_2$$
and 
\begin{equation*}
\begin{aligned}
\int_\Real |\hat{g}(\xi_1)|^2 &(1 + |\xi_1|^2 + |\xi_1|^4)^{\frac{1}{2}} \der \xi_1 = \frac{1}{4 \pi^2} \int_\Real \mid\int_{\Real} \hat{u}(\xi) \der \xi_2\mid^2 (1 + |\xi_1|^2 + |\xi_1|^4)^{\frac{1}{2}} \der \xi_1 \,,\\
&\leq  \frac{1}{4 \pi^2} \int_\Real (1 + |\xi_1|^2 + |\xi_1|^4)^{\frac{1}{2}} \int_{\Real} |\hat{u}(\xi)|^2  (1 + |\xi|^2 + |\xi_1|^4)\der \xi_2   \int_\Real  (1 + |\xi|^2 + |\xi_1|^4)^{-1}       \der \xi_2 \der \xi_1 \,.\\
\end{aligned}
\end{equation*}
Moreover, 
$$  \int_\Real  (1 + |\xi|^2 + |\xi_1|^4)^{-1}       \der \xi_2  = \frac{\pi}{(1 + |\xi_1|^2 + |\xi_1|^4)^{\frac{1}{2}}}.    $$
Thus, 
$$\int_\Real |\hat{g}(\xi_1)|^2 (1 + |\xi_1|^2 + |\xi_1|^4)^{\frac{1}{2}} \der \xi_1
\leq  \frac{1}{4 \pi} \int_{\Real^2} |\hat{u}(\xi)|^2  (1 + |\xi|^2 + |\xi_1|^4)\der \xi .$$
Equivalently, we obtain 
$$\|u(x,0)\|_{\mathtt{H}^{\frac{1}{2}}} \leq C \|u\|_{\mathtt{H}^1(\Real^2_+)} .$$

Reversely, for each $g \in \mathtt{H}^\frac{1}{2}$ we construct a function 
$$ \hat{u}(\xi) := 2 \,\frac{\hat{g_1}(\xi_1)}{\xi_2^2 + (1 + |\xi_1|^2 + |\xi_1|^4)} \,, $$
where 
$$ \hat{g_1}(\xi_1) = \hat{g}(\xi_1)(1 + |\xi_1|^2 + |\xi_1|^4)^\frac{1}{2} \,.$$
We have
\begin{equation*}
\begin{aligned}
\int_{\Real^2} |\hat{u}(\xi)|^2 (1+ |\xi|^2 + |\xi_1|^4) \der \xi &= 4 \int_{\Real^2}  \frac{|\hat{g_1}(\xi_1)|^2}{\xi_2^2+ 1+ |\xi_1|^2 + |\xi_1|^4} \der \xi_1 \der \xi_2 \,,\\
&\leq 4 \int_{\Real} |\hat{g_1}(\xi_1)|^2 \frac{\pi}{(1 + |\xi_1|^2 + |\xi_1|^4)^\frac{1}{2}} \der\xi_1 \,,\\
&\leq 4\pi \int_\Real |\hat{g}(\xi_1)|^2 (1 + |\xi_1|^2 + |\xi_1|^4)^\frac{1}{2} \der\xi_1 \,,\\
&\leq 4\pi \|g\|^2_{\mathtt{H}^\frac{1}{2}} \,.
\end{aligned}
\end{equation*}
Therefore, $u(x,y) \in \mathtt{H}^1(\Real^2_+)$ for $y \geq 0$ and and $2 \sqrt{\pi}\|g(x)\|_{\mathtt{H}^{\frac{1}{2}}} \geq  2 \pi\|u\|_{\mathtt{H}^1(\Real^2_+)}$.

Moreover, we take $u_n \in C^\infty(\overline{\Real^2}; \C^2): u_n  \xrightarrow[n \to \infty]{} u$ in $\Dom(T_{\Real^2})$ then $\gamma u_n(x,0) \xrightarrow[n \to \infty]{} \gamma u(x,0)$ in $\mathtt{H}^\frac{1}{2}(\Real)$ and $\hat{u}_n(\xi) \xrightarrow[n \to \infty]{} \hat{u}(\xi)$ in
$L^2(\Real^2)$. Using the inverse Fourier transform, we estimate
$$u_n(x,y) = \frac{1}{(2\pi)^2} \int_{\Real^2} e^{i x\xi_1 + i y \xi_2} \hat{u}_n(\xi) \der \xi.$$
Thus,
\begin{equation}
\begin{aligned}
u_n(x,0) &= \frac{1}{4\pi^2} \int_{\Real^2} e^{i x\xi_1 } \hat{u}_n(\xi) \der \xi \,,\\
&\xrightarrow[n \to \infty]{}  \frac{1}{4\pi^2} \int_{\Real^2} e^{i x\xi_1 } \hat{u}(\xi) \der \xi         
= \frac{1}{2\pi^2} \int_{\Real^2} e^{i x\xi_1 } \frac{\hat{g_1}(\xi_1)}{\xi_2^2 + (1 + |\xi_1|^2 + |\xi_1|^4)} \der \xi         
\,,\\
&= \frac{1}{2\pi}\int_{\Real} e^{i x\xi_1 } \hat{g}(\xi_1) \der \xi_1 = g(x) \,.
\end{aligned}
\end{equation}
Therefore, $\gamma u(x,0) = g(x) \in \mathtt{H}^\frac{1}{2}(\Real)$.
It concludes the proof.
\end{proof}

\begin{Theorem}\label{trace}
There exists a continuous trace map $\gamma : \mathcal{M} \rightarrow \mathtt{H}_{\C^2}^{-\frac{1}{2}}$.

\end{Theorem}
\begin{proof}
For each  $u \in \mathcal{M}$, there exists a sequence $\{u_n\}_{n\geq1} \in C^\infty(\overline{\Real^2_+})$ such that $ u_n \xrightarrow[n \to \infty]{} u$ in $\mathcal{M}$ due to Lemma \ref{denseM}. We will show that the traces $\gamma u_n$ converges in $\mathtt{H}_{\C^2}^{-\frac{1}{2}}$.

Take $g = \begin{psmallmatrix} g_1 \\ g_2 \end{psmallmatrix} \in \mathtt{H}_{\C^2}^{\frac{1}{2}}$, there exists a function $v = \begin{psmallmatrix} v_1 \\ v_2 \end{psmallmatrix} \in \mathtt{H}_{\C^2}^1$ such that $ \|v\|_{\mathtt{H}_{\C^2}^1}\leq 4 \|g\|_{\mathtt{H}_{\C^2}^{\frac{1}{2}}},  \gamma v = g$ by virtue of Theorem \ref{trace0}. We estimate
\begin{equation}
\begin{aligned}
(u_n, Tv) &= \int_0^\infty \int_{\Real} \overline{u_{1n}} [-i \partial_y v_1 + (-\partial^2_x + \delta) v_2]  + \overline{u_{2n}} [(-\partial^2_x + \delta)v_1 + i \partial_y v_2 ]\ \der x \der y \,,\\
&= (Tu_n,v) + \int_\Real i\overline{u_{1n}} v_1|_{y=0} - i\overline{u_{2n}} v_2|_{y=0} \ \der x \,. \\
\end{aligned}
\end{equation}
Therefore, 
\begin{equation}
\begin{aligned}
(u_n - u_m, Tv) &=  (Tu_n-Tu_m,v) + \int_\Real i\overline{(u_{1n}- u_{1m})} v_1|_{y=0} - i\overline{(u_{2n}- u_{2m})} v_2|_{y=0} \ \der x  \\
\end{aligned}
\end{equation}
for all $m, n \geq 1$.

Now we choose $g_2=v_2=0$. It deduces that 
\begin{equation*}
\begin{aligned}
|(u_{1n}(x,0) - u_{1m}(x,0), v)| &\leq \|u_m-u_n\| \|Tv\| + \|Tu_n -Tu_m\| \|v\| \,,\\
&\quad\leq C (\|u_n-u_m\| + \|Tu_n-Tu_m\|) \|v\|_{\mathtt{H}_{\C^2}^1} \,,\\
&\quad\leq M (\|u_n-u_m\| + \|Tu_n-Tu_m\|)\|v\|_{\mathtt{H}_{\C^2}^\frac{1}{2}} \,,
\end{aligned}
\end{equation*}
with $M, C$ are constants being independent of $v$.
Hence,
$\lim_{n\rightarrow\infty} (u_{1n}- u_{1m})|_{y=0} =0$ exists in 
$\mathtt{H}^{-\frac{1}{2}}$. Analogously, if we choose $g_1=v_1=0$, we also obtain $\lim_{n\rightarrow\infty} (u_{2n}- u_{2m})|_{y=0}=0$ exists in 
$\mathtt{H}^{-\frac{1}{2}}$. As a result, there exists a  trace map 
\begin{equation*}
\begin{aligned}
 \gamma : \quad \quad  &\mathcal{M} \longrightarrow \mathtt{H}_{\C^2}^{-\frac{1}{2}} \\
          &u \longmapsto  \gamma u = \lim_{n\rightarrow\infty} u_n(x,0) \,.
\end{aligned}
\end{equation*}
The limit is independent of the chosen convergent sequences. It means that if $\{u_n\}_{n\geq1}, \{u'_n\}_{n\geq 1}$ simultaneously converge to $u$ in $\mathcal{M}$ then $\gamma u$ is still uniquely identified.

\end{proof}

\begin{Theorem}\label{half-space self-adjoint}
$T_{\Real^2_+}$ is a self-adjoint operator.
\end{Theorem}
\begin{proof}
For all $u,v \in \Dom(T_{\Real^2_+})$, we have
\begin{equation}\label{symmetric}
\begin{aligned}
(u, T_{\Real^2_+}v) &= \int_0^\infty \int_{\Real} \overline{u_1} [-i \partial_y v_1 + (-\partial^2_x + \delta) v_2]  + \overline{u_2} [(-\partial^2_x + \delta)v_1 + i \partial_y v_2 ]\ \der x \der y \,,\\
&= \int_0^\infty \int_{\Real} i \partial_y\overline{u_1} v_1 + (-\partial^2_x + \delta) \overline{u_1} v_2 \ \der x \der y + \int_\Real i\overline{u_1} v_1|_{y=0} \ \der x  \\
& \quad + \int_0^\infty \int_{\Real} -i \partial_y\overline{u_2} v_1 + (-\partial^2_x + \delta) \overline{u_2} v_1\ \der x \der y + \int_\Real -i\overline{u_2} v_2|_{y=0} \ \der x \,,\\
&= (T_{\Real^2_+}u,v) \,.
\end{aligned}
\end{equation}
It implies that $T_{\Real^2_+}$ is symmetric due to the fact that $u_1 = u_2, v_1 = v_2$ on $y=0$. 

Now we take $u \in \Dom(T^*_{\Real^2_+})$, recall the definition for the domain of the adjoint operator
\begin{equation}\label{operator.1D}
\begin{aligned}
\Dom(T^*_{\Real^2_+}) &:= \left\{
  u \in L^{2}(\Real^2_+) :  \mbox{ there exists } \phi_u \in L^2(\Real^2_+), \forall v \in \Dom(T_{\Real^2_+}), (u, T_{\Real^2_+}v) = (\phi_u, v)
  \right\}
  \,.
\end{aligned}
\end{equation}
For each $w = \begin{psmallmatrix}
w_1 \\ w_2
\end{psmallmatrix}$, we denote $w^{*} := \begin{psmallmatrix}
w_2 \\ w_1
\end{psmallmatrix}$.
Taking $v \in C^\infty_0(\Real^2_+)$ we have
\begin{equation}\label{symmetric}
\begin{aligned}
(u, T_{\Real^2_+}v) &= \int_0^\infty \int_{\Real} \overline{u_1} [-i \partial_y v_1 + (-\partial^2_x + \delta) v_2]  + \overline{u_2} [(-\partial^2_x + \delta)v_1 + i \partial_y v_2 ]\ \der x \der y \,,\\
&= \left\langle Tu,v \right\rangle = (\phi_u,v) \,,
\end{aligned}
\end{equation}
here $\langle\cdot,\cdot\rangle$ is the duality bracket of distributions. It deduces that $Tu \in L^2(\Real^2_+)$ and thus, $Tu = \phi_u$.

For all $ v \in \Dom(T_{\Real^2_+})$, we have
\begin{equation}
\begin{aligned}
(u, T_{\Real^2_+}v) &= \int_0^\infty \int_{\Real} \overline{u_1} [-i \partial_y v_1 + (-\partial^2_x + \delta) v_2]  + \overline{u_2} [(-\partial^2_x + \delta)v_1 + i \partial_y v_2 ]\ \der x \der y \,,\\
&= \int_0^\infty \int_{\Real} i \partial_y\overline{u_1} v_1 + (-\partial^2_x + \delta) \overline{u_1} v_2 \ \der x \der y + \int_\Real i\overline{u_1} v_1|_{y=0} \ \der x  \\
& \quad + \int_0^\infty \int_{\Real} -i \partial_y\overline{u_2} v_1 + (-\partial^2_x + \delta) \overline{u_2} v_1\ \der x \der y + \int_\Real -i\overline{u_2} v_2|_{y=0} \ \der x \,,\\
&= (Tu,v) + \int_\Real i\overline{(u_1-u_2)} v_1|_{y=0} \ \der x   \,.
\end{aligned}
\end{equation}
It yields that 
$ u_1 = u_2 $ in $\mathtt{H}^{-\frac{1}{2}}$.

Now we  extend the functions $u$ to $\Real^2$ 
\begin{equation}
\tilde{u}(x,y) = \left\{\begin{array}{rcl}
	u(x,y) &\mbox{ if } y \geq 0  \,,\\
	 u(x,-y)^* &\mbox{ if } y < 0  \,,\\
\end{array}\right.
\end{equation}

For all $v \in C^\infty_0(\Real^2; \C^2)$, we put $w(x,y)= v(x,-y)$ for all $(x,y) \in \Real^2$ then we deduce that $w \in C^\infty_0(\Real^2)$. Taking $u_1=u_2$ in  $\mathtt{H}^{-\frac{1}{2}}$ into account, we have
\begin{equation*}
\begin{aligned}
|(\tilde{u}, Tv)| &= |\int_0^\infty \int_{\Real} \overline{u_1} [-i \partial_y v_1 + (-\partial^2_x + \delta) v_2]  + \overline{u_2} [(-\partial^2_x + \delta)v_1 + i \partial_y v_2 ]\ \der x \der y \\
&\quad+ \int_{-\infty}^0 \int_{\Real} \overline{\tilde{u}_1} [-i \partial_y v_1 + (-\partial^2_x + \delta) v_2]  + \overline{\tilde{u}_2} [(-\partial^2_x + \delta)v_1 + i \partial_y v_2 ]\ \der x \der y |\,,\\
&= |(Tu,v)_{L^2(\Real^2_+)} + \int_\Real i(\overline{u_1}v_1-\overline{u}_2 v_2)|_{y=0} \ \der x  + \int_{-\infty}^0 \int_{\Real} i \partial_y\overline{\tilde{u}_1} v_1 + (-\partial^2_x + \delta) \overline{\tilde{u}_1} v_2 \ \der x \der y \\
&\quad+ \int_\Real -i\overline{\tilde{u}_1} v_1|_{y=0} \ \der x   + \int_{-\infty}^0 \int_{\Real} -i \partial_y\overline{\tilde{u}_2} v_1 + (-\partial^2_x + \delta) \overline{\tilde{u}_2} v_1\ \der x \der y + \int_\Real i\overline{\tilde{u}_2} v_2|_{y=0} \ \der x |\,,\\
&= | (Tu,v)_{L^2(\Real^2_+)} + \int_\Real i(\overline{u_1}v_1-\overline{u}_2 v_2)|_{y=0} \ \der x + ((Tu)^*,w)_{L^2(\Real^2_+)} + \int_\Real -i(\overline{u_2}v_1-\overline{u}_1 v_2)|_{y=0} \ \der x |\,,\\
&= |(Tu,v)_{L^2(\Real^2_+)} + (Tu,w^*)_{L^2(\Real^2_+)}| \leq \|Tu\|_{L^2(\Real^2_+)} \|v\|_{L^2(\Real^2)} \,.
\end{aligned}
\end{equation*}
It implies that $\tilde{u} \in \Dom(T^*_{\Real^2})$. Moreover, $\Dom(T^*_{\Real^2})= \Dom(T_{\Real^2})$ due to the self-adjointess of $T_{\Real^2}$. Hence, we deduce that $u \in H^1(\Real^2_+)$ and $\partial_x^2 u \in L^2(\Real^2_+)$. Besides, we also have $u^*(x,0) = u(x,0)$ due to the trace theorem and thus, it deduces that $ u_1 = u_2$ on $y=0$. This concludes the proof of the theorem.

\end{proof}
%

\section{The spectral properties of semi-Dirac semi-metals}\label{Sec.proof}
%
By virtue of the self-adjointness,
 the closest-to-zero squared spectrum of the operator $T_{\Real^2_+}$ 
can be computed by

\begin{equation*}
  \lambda_1(\Real^2_+)^2 
  = \inf_{\stackrel[u \not= 0]{}{u \in \Dom(T_{\Real^2_+})}} 
  \frac{\|T_{\Real^2_+}u\|^2}{\, \|u\|^2}  
\end{equation*}
where $$\|T_{\Real^2_+} u\|^2   = \|\partial_y u\|^2 + \|\partial_x^2 u\|^2 + 2 \delta\|\partial_x u\|^2 
  + \delta^2 \|u\|^2.$$
Hence, 
\begin{equation*}
  \lambda_1(\Real^2_+)^2 - \delta^2
  = \inf_{\stackrel[u \not= 0]{}{u \in \Dom(T_{\Real^2_+})}} 
  \frac{\|\partial_y u\|^2 + \|\partial_x^2 u\|^2 + 2 \delta\|\partial_x u\|^2 }{\, \|u\|^2}  \geq 0
\end{equation*}
Thus, the spectrum of $T_{\Real^2_+}$: $\sigma_{T_{\Real^2_+}} \subset (-\infty, -\delta] \cup [\delta, \infty)$. Moreover, this spectrum is purely essential, we have the following lemma
\begin{Lemma}\label{pointsp}
The point spectrum of $T_{\Real^2_+}$ is empty.
\end{Lemma}
\begin{proof}
Suppose that the point spectrum of $T_{\Real^2_+}$ is not empty then there exists $ \lambda \in \Real, u = \begin{psmallmatrix}
u_1 \\ u_2 
\end{psmallmatrix} \in \Dom(T_{\Real^2_+})$ being the eigenvalue and eigenfunction of $T_{\Real^2_+}$ respectively, we have
\begin{equation}\label{system1} 
\left\{
\begin{aligned}
  -i\partial_y  u_1 + (-\partial_x^2 + \delta) u_2 &= \lambda u_1 \quad && \mbox{in} \quad \Real^2_+ \,,\\
(-\partial_x^2 + \delta) u_1 + i\partial_y  u_2   &= \lambda  u_2  \quad && \mbox{in} \quad \Real^2_+ \,,\\
 u_1 &= u_2 \quad && \mbox{on } y=0 .
\end{aligned}  
\right.
\end{equation}
Extending $u$ to $\Real^2$, we put
\begin{equation}
\tilde{u}(x,y) = \left\{\begin{array}{rcl}
	u(x,y) &\mbox{ if } y \geq 0  \,,\\
	 u^*(x,-y) &\mbox{ if } y < 0  \,.\\
\end{array}\right.
\end{equation}
It yields
$$
\partial_x\tilde{u}(x,y) = \left\{\begin{array}{rcl}
	\partial_x u(x,y) &\mbox{ if } y \geq 0  \,,\\
	 \partial_x u^*(x,-y) &\mbox{ if } y < 0  \,,\\
\end{array}\right.
\quad \quad \quad 
\partial_x^2\tilde{u}(x,y) = \left\{\begin{array}{rcl}
	\partial_x^2 u(x,y) &\mbox{ if } y \geq 0  \,,\\
	 \partial_x^2 u^*(x,-y) &\mbox{ if } y < 0  \,,\\
\end{array}\right.
$$
and $$\partial_y\tilde{u}(x,y) = \left\{\begin{array}{rcl}
	\partial_y u(x,y) &\mbox{ if } y \geq 0  \,,\\
	 -\partial_y u^*(x,-y) &\mbox{ if } y < 0  \,.\\
\end{array}\right.
$$
It deduces that $\tilde{u} \in H^1(\Real^2), \partial^2_x \tilde{u} \in L^2(\Real^2)$ and $\tilde{u}$ is an eigenfunction  corresponding to an eigenvalue $\lambda$ of $T_{\Real^2}$. As a result, the point spectrum of $T_{\Real^2}$ is not empty. It is a contrary to the fact that the spectrum of $T_{\Real^2}$ is purely absolutely continuous and the point spectrum of $T_{\Real^2}$ is empty by means of the Fourier transform. Consequently, the point spectrum of $T_{\Real^2_+}$ is empty. The proof is completed.

\end{proof}
Regarding Lemma \ref{pointsp} and the self-adjointness of $T_{\Real^2_+}$, it is apparent that the spectrum of $T_{\Real^2_+}$ is purely essential : $\sigma(T_{\Real^2_+}) = \sigma_{\textup{ess}}(T_{\Real^2_+})$. Now we compute the spectrum of $T_{\Real^2_+}$ as follows,
\begin{Theorem}
$$\sigma_{T_{\Real^2_+}} =(-\infty, -\delta] \cup [\delta, \infty).$$
\end{Theorem}
\begin{proof}
Recall that $\sigma_{T_{\Real^2_+}} \subset (-\infty, -\delta] \cup [\delta, \infty)$ due to the min-max principle. Therefore, we only have to verify the opposite.

Take $k \in \Real, k\geq 0$ and choose a function $\varphi(x,y) \in C^\infty_0(\Real^2_+)$ such that $\|\varphi\| = \frac{1}{\sqrt{2}}$. Putting
$$ \varphi_n(x,y): = \frac{1}{n} \varphi(\frac{x}{n}, \frac{y}{n})$$
then we also have $\|\varphi_n\| =\frac{1}{\sqrt{2}}$. Moreover,
$$ \partial_x \varphi_n(x,y) = \frac{1}{n^2} \partial_x\varphi(\frac{x}{n}, \frac{y}{n}), \quad \quad \quad \partial_y \varphi_n(x,y) = \frac{1}{n^2} \partial_y\varphi(\frac{x}{n}, \frac{y}{n}) \,,$$
$$\partial_x^2 \varphi_n(x,y) = \frac{1}{n^3} \partial_x^2\varphi(\frac{x}{n}, \frac{y}{n}) \,.$$
Therefore,
$$ \|\partial_x \varphi_n(x,y)\| = \frac{1}{n} \|\partial_x\varphi\|, \quad \|\partial_y \varphi_n(x,y)\| = \frac{1}{n} \|\partial_y\varphi\|, \quad \|\partial_x^2 \varphi_n(x,y)\| = \frac{1}{n^2} \|\partial_x^2\varphi\| \,.
$$
Putting 
$$ u_n(x,y) := \varphi_n(x,y) e^{ikx}, \quad \quad \psi_n := \begin{psmallmatrix} u_n \\ u_n \end{psmallmatrix}$$
then $\psi_n  \in \Dom(T_{\Real^2_+}) , \|\psi_n\|=1$ and 
$$ \partial_x u_n = \partial_x \varphi_n e^{ikx} + ik \varphi_n e^{ikx}, \quad \quad 
\partial_x^2 u_n = \partial_x^2 \varphi_n e^{ikx} + 2ik \, \partial_x\varphi_n e^{ikx} - k^2 u_n\,, $$
$$\partial_y u_n = \partial_y \varphi_n e^{ikx} \,.$$

For every $\mu \in [\delta, \infty)$, we choose $k^2 = \mu - \delta$ and
 estimate
\begin{equation}
\begin{aligned}
\|T_{\Real^2_+} \psi_n - \mu \psi_n\|^2 &= \| -i \partial_y u_n -\mu u_n + (-\partial_x^2 + \delta) u_n \|^2 + \| i \partial_y u_n -\mu u_n + (-\partial_x^2 + \delta) u_n \|^2 \,,\\
&= \| -i \partial_y u_n  -\partial_x^2 \varphi_n e^{ikx} - 2ik \, \partial_x\varphi_n e^{ikx} \|^2  + \| i \partial_y u_n  -\partial_x^2 \varphi_n e^{ikx} - 2ik \, \partial_x\varphi_n e^{ikx} \|^2 \,,\\
&\leq 2 (\|\partial_y \varphi_n\|^2 + \|\partial_x^2\varphi_n\|^2 + \|2k \, \partial_x\varphi_n\|^2) \xrightarrow[n \to \infty]{} 0 \,.
\end{aligned}
\end{equation}
It follows that $\mu \in \sigma(T_{\Real^2_+})$ due to the Weyl theorem.

On the other hand, for every $\mu \in (-\infty, -\delta]$, we choose $k^2 = -\delta-\mu$ and $\phi_n = \begin{psmallmatrix} u_n \\ -u_n \end{psmallmatrix} \in \Dom(T_{\Real^2_+})$. Employing the analogous estimates, we also obtain that 
$$\|T_{\Real^2_+} \phi_n - \mu \phi_n\| \xrightarrow[n \to \infty]{} 0 \,.$$
Equivalently, $\mu \in \sigma(T_{\Real^2_+})$. Hence, it concludes the proof.

\end{proof}
Summarily, we have investigated that $\sigma(T_{\Real^2_+}) = \sigma_{\textup{ess}}(T_{\Real^2_+}) = (-\infty, -\delta] \cup [\delta, \infty)$. In order to establish the point spectrum of the semi-Dirac semi-metals, we need to add some potentials.

Let us take  $V \in L^\infty(\Real^2_+; \Real)$, we consider the following operator
\begin{equation}
\begin{aligned}
H := T_{\Real^2_+} + V \begin{pmatrix}
0 & 1 \\
1 & 0
\end{pmatrix} \, 
 \mbox{ with } \Dom(H):= \Dom(T_{\Real^2_+}) \,.
\end{aligned}
\end{equation}
By the boundedness of $V$ and the self-adjointness of $T_{\Real^2_+}$ then it is apparent that $H$ is also self-adjoint. Now we illustrate the effect of the sign of $V$ on the properties of the spectrum of $H$.
\begin{Proposition}
If $L^\infty(\Real^2_+; \Real) \ni V= V(x) \geq 0$ a.e. in $\Real^2_+$ then $\sigma(H) \subset (-\infty, -\delta] \cup [\delta, \infty)$.
\end{Proposition}
\begin{proof}
By the hypothesis,  for any $u \in \Dom(H)$, we have
\begin{equation*}
\begin{aligned}
\|Hu\|^2 &= \|\Big(T_{\Real^2_+}+ V \begin{pmatrix}
0 & 1 \\
1 & 0
\end{pmatrix}\Big) u \|^2 \\
&= \|-i \partial_y u_{1} + (-\partial^2_x + \delta + V) u_{2}\|^2 + \|(-\partial^2_x + \delta + V)u_{1} + i \partial_y u_{2}\|^2 \,,\\
&= \|-i \partial_y u_{1}\|^2 + \|(-\partial^2_x + \delta + V) u_{2}\|^2 + 2\Re (-i \partial_y u_1, (-\partial^2_x + \delta + V) u_2)\\
& \quad + \|(-\partial^2_x + \delta + V)u_{1}\|^2 + \|i \partial_y u_{2}\|^2 
+ 2 \Re ((-\partial^2_x + \delta + V)u_{1}, i \partial_y u_{2}) \,,\\
&= \|\partial_y u\|^2 + \|(-\partial^2_x + \delta + V) u\|^2 \geq \delta^2 \|u\|^2\,.
\end{aligned}
\end{equation*}
It follows that $\sigma(H) \subset  (-\infty, -\delta] \cup [\delta, \infty)$ due to the min-max principle. 
This completes the proof.

\end{proof}

In order to establish sufficient conditions for the potentials to guarantee the existence of the discrete  spectrum of the operator, we introduce the following real vector space
$$L^\infty_0(\Real^2_+; \Real) := \{ u \in L^\infty (\Real^2_+; \Real) : u \mbox{ has compact support in  }  \Real^2_+   \} \,.
$$

\begin{Theorem}\label{discrete1}
There exists $V = V(x,y) \in L^\infty_0(\Real^2_+; \Real)$ such that the discrete spectrum of $H$ is not empty.
\end{Theorem}
\begin{proof}
First of all, it is striking that $\sigma_{\textup{ess}}(H)= \sigma_{\textup{ess}}(T_{\Real^2_+})= (-\infty, -\delta] \cup [\delta, \infty)$ due to the boundedness and vanishing  at infinity of $V$ as discussed in \cite{KA,Reed,T92}.

We study on the Dirichlet Laplacian defined on a segment $[a,b]$ with $ 0 < a < b$. Denote by $(\lambda_1, v_1(x))$, $(\lambda_3, v_3(x))$, $(\lambda_1, u_2(y))$  the ground-state, the third eigenvalues and  eigenfunctions and the first eigenpair of the Dirichlet Laplacian with respect to variables $x, y$ determined on $[a,b]$ such that $\|v_1(x)\|_{L^2([a,b])} = \|v_3(x)\|_{L^2([a,b])} = \|u_2(y)\|_{L^2([a,b])} = 1$. As a result, $\lambda_1 = \frac{\pi^2}{(b-a)^2}, \lambda_3 = \frac{9\pi^2}{(b-a)^2} $. Thoroughly, $v_1(x), u_2(y)$ satisfy the following system with respect to $z=x$ and $z=y$
\begin{equation}\label{Dirichlet} 
\left\{
\begin{aligned}
  - \partial_z^2 u(z) &= \lambda_1 u(z)
  && \mbox{in} && [a,b] 
  \,,
  \\
  u(a) &= u(b) = 0 
  \,,
\end{aligned}
\right.
\end{equation}
and we also have
\begin{equation} 
\left\{
\begin{aligned}
  - v_3''(x) &= \lambda_3 v_3(x)
  && \mbox{in} && [a,b] 
  \,,
  \\
  v_3(a) &= v_3(b) = 0 
  \,.
\end{aligned}
\right.
\end{equation}
Moreover, we can select 
$$ v_1(x) = \sqrt{\frac{2}{b-a}} \sin \frac{\pi (x-a)}{b-a} \,, \quad\quad\quad \quad \quad\quad v_3(x) = \sqrt{\frac{2}{b-a}} \sin \frac{3\pi (x-a)}{b-a} \,.
$$
 Putting 
 $$ u_1(x) := \frac{-3}{\sqrt{10}} v_1(x) + \frac{1}{\sqrt{10}} v_3(x) \,.$$
It yields that
$\|u_1(x)\|_{L^2([a,b])} =1, u_1(a)= u_1(b) = 0$ and
\begin{equation}
\begin{aligned}
u_1'(x) &=  \frac{-3}{\sqrt{10}} v_1'(x) + \frac{1}{\sqrt{10}} v_3'(x) \,,\\
&= \frac{-3}{\sqrt{10}} \, \sqrt{\frac{2}{b-a}} \, \frac{\pi}{b-a} \, \cos \frac{\pi (x-a)}{b-a} + \frac{1}{\sqrt{10}} \, \sqrt{\frac{2}{b-a}} \, \frac{3\pi}{b-a} \, \cos \frac{3\pi (x-a)}{b-a} \,.\\
\end{aligned}
\end{equation}
As a result, we deduce that $$ u_1'(a) = \frac{-3}{\sqrt{10}} \, \sqrt{\frac{2}{b-a}} \, \frac{\pi}{b-a}  + \frac{1}{\sqrt{10}} \, \sqrt{\frac{2}{b-a}} \, \frac{3\pi}{b-a} = 0 \,,$$
$$ u_1'(b) = \frac{-3}{\sqrt{10}} \, \sqrt{\frac{2}{b-a}} \, \frac{\pi}{b-a} \, (-1) + \frac{1}{\sqrt{10}} \, \sqrt{\frac{2}{b-a}} \, \frac{3\pi}{b-a} \, (-1) = 0 \,.$$
Extending $u_1(x)$ to $\Real$
\begin{equation}
\tilde{u}_1(x) = \left\{\begin{array}{rcl}
	&u_1(x) &\mbox{ if } x \in [a,b]  \,,\\
	 &0     &\mbox{ if } x \in (-\infty, a) \cup (b, \infty)  \,.\\
\end{array}\right.
\end{equation}
It is apparent that $\tilde{u}_1 \in H^1(\Real)$ due to the fact that the smooth function $u_1(x)$ vanishes at $a$ and $b$. Furthermore,  for all $\varphi \in C^\infty_0(\Real)$, we have
\begin{equation}
\begin{aligned}
|(\tilde{u}'_1, \varphi')_{L^2(\Real)}| &= |\int_{[a,b]} u'_1(x) \varphi'(x) \, \der x |
= |\int_{[a,b]} - u_1''(x) \varphi (x) \, \der x  + (u_1'\varphi)|_a^b \, | \,,\\
&= |\int_{[a,b]} - u_1''(x) \varphi (x) \, \der x | \leq \|u_1''\|_{L^2([a,b])} \|\varphi\|_{L^2(\Real)} \,.
\end{aligned}
\end{equation}
It follows that $\tilde{u}_1''(x) \in L^2(\Real)$. In particular,
\begin{equation}
\begin{aligned}
u''_1(x) &= \frac{-3}{\sqrt{10}} v_1''(x) + \frac{1}{\sqrt{10}} v_3''(x) 
=  \frac{3}{\sqrt{10}} \lambda_1 v_1 - \frac{1}{\sqrt{10}} \lambda_3 v_3  \,.\\
\end{aligned}
\end{equation}

Putting $\psi(x,y) := u_1(x) u_2(y), v := \begin{psmallmatrix} \psi \\ \psi \end{psmallmatrix}$. In order to reduce notations, we can identity $v$ with the trivial extension of $v$ to $\Real^2$ and thus, we have $ v \in \Dom(H)$. It is obvious that $\|\psi\|=1, \|v\|= \sqrt{2}$.

Now we choose $V \in L^\infty_0(\Real^2_+; \Real)$ such that $V$ is a constant in $[a,b]\times [a,b]$, we compute
\begin{equation}
\begin{aligned}
\|Hv\|^2 &= 2\|\partial_y\psi\|^2 + 2\|(-\partial_x^2 + \delta + V) \psi\|^2 \,,\\
&= 2\|u_1(x)\|_{L^2(a,b)}^2 \|\partial_y u_2\|_{L^2(a,b)}^2 +  2\|u_2(y)\|^2_{L^2(a,b)} \|u_1''(x) + (\delta + V)u_1\|_{L^2(a,b)}^2 \,,\\
&= 2\lambda_1 + 2\|(\frac{3}{\sqrt{10}} \lambda_1 v_1 - \frac{1}{\sqrt{10}} \lambda_3 v_3) + (\delta + V)(\frac{-3}{\sqrt{10}} v_1 + \frac{1}{\sqrt{10}} v_3)\|_{L^2(a,b)}^2 \,,\\
&= 2 \lambda_1 + 2 \large\left( \frac{9}{10}\lambda_1^2 + \frac{1}{10} \lambda_3^2 + (\delta +V)^2 - (\delta + V) (\frac{9}{5}\lambda_1 + \frac{1}{5} \lambda_3)  \large\right) \,,\\
&= 2 \lambda_1 + 2 \large\left( 9 \lambda_1^2 + (\delta + V)^2  -\frac{18}{5} \lambda_1 (\delta + V)       \large\right) \,.
\end{aligned}
\end{equation}

 It yields that
\begin{equation}
\begin{aligned}
\|Hv\|^2 - \delta^2 \|v\|^2 &= 2\lambda_1 + 18\lambda_1^2 - \frac{36}{5} \delta\lambda_1 -\frac{36}{5}  \lambda_1 V + 4 \delta V + 2 V^2 \ \,,\\
&= 2V^2 + 4(-\frac{9}{5}\lambda_1 + \delta) V + 2 \lambda_1 -\frac{36}{5}  \delta\lambda_1 + 18 \lambda_1^2 \,.
\end{aligned}
\end{equation}
The reduced discriminant of this quadratic trinomial equals $$ 4(-\frac{9}{5}\lambda_1 + \delta)^2 - 4 \lambda_1 + \frac{72}{5}  \lambda_1\delta - 36 \lambda_1^2 = 4 \delta^2 - 4 \lambda_1 - \frac{576}{25} \lambda_1^2.$$
If
\begin{equation}\label{condition}
 \lambda_1 + \frac{576}{100} \, \lambda_1^2  < \frac{\pi^2}{(b-a)^2} + 6 \, \frac{\pi^4}{(b-a)^4}   < \delta^2 \,
\end{equation} 
  then the corresponding quadratic equation has two solutions $V_1, V_2 : V_1 < V_2$. Therefore, for any $V \in (V_1, V_2)$ in $[a,b]^2$ we have  $\|Hv\|^2 - \delta^2 \|v\|^2 < 0$. It deduces the existence of the discrete spectrum of $H^2$ as well as $H$. It concludes the proof.

\end{proof}
This result provides us a local selection for the potential $V$ defined in an arbitrary square in $\Real_+ \times \Real_+$ with the aim of establishing the existence of the discrete spectrum of the semi-Dirac semi-Laplacian problem in the half-space. It can be seen that the rigorous relation between $\delta$ and $b-a$ described in \eqref{condition} can be supplemented by a numerical evidence.
Figure \ref{condi} illustrates the validity region of \eqref{condition}.
\begin{figure}[h]
  \begin{center}
  \includegraphics[scale=0.8]{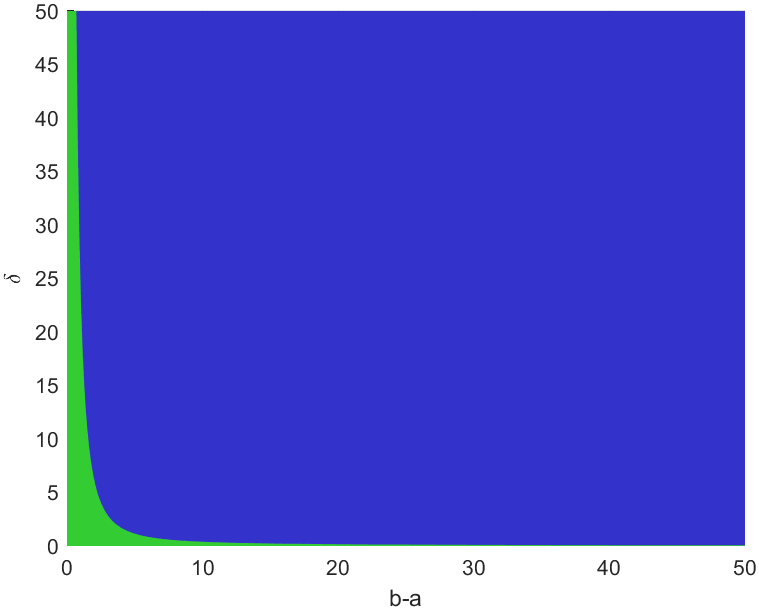}
  \caption{The blue colour indicates the validity region of \eqref{condition}.
  }
  \label{condi}
  \end{center}
\end{figure}
\begin{Remark}
Based on our analysis, it is apparent that the local establishment for $V$ can be alternatively determined in an appropriate rectangle in $\Real^2_+$ by virtue of employing the separation of variables for trial functions. Nevertheless, the setting of that in squares in $\Real_+ \times \Real_+$ perhaps completely fulfills our principle purposes.
\end{Remark}

Now we establish perturbed operators  defined in \cite{KA} to analyse the spectral stability in small perturbations, we denote 
$$W := \begin{pmatrix}
w_{11} & w_{12} \\
w_{21} & w_{22}
\end{pmatrix} \,, \quad \quad  H_\epsilon := T_{\Real^2_+} + \epsilon W ,,$$
where $w_{12}, w_{21} : \Real^2_+ \rightarrow \C$ in $L^\infty(\Real^2_+)$, $w_{ii} \in L^\infty(\Real^2_+; \Real)$ for  $ i = 1,2$, $w_{ij}$ vanishes at infinity for every $ i, j = 1,2$ and $\Dom(H_\epsilon):= \Dom(T_{\Real^2_+})$.

Recall that if $w_{12}= \overline{w_{21}}$ then $H_\epsilon$ is self-adjoint \cite{KA} and $\sigma_{\textup{ess}}(H_\epsilon) = \sigma_{\textup{ess}}(T_{\Real^2_+})= (-\infty, -\delta] \cup [\delta, \infty)$ due to the boundedness and vanishing  at infinity of $W$. This result can be seen as the stability of the essential spectrum for the perturbed operators.

Considering  a sequence of approximate functions of $1$ in $L^2(\Real^2_+)$ as analogously defined in \cite{KA}, we choose $g(t) \in C^\infty([0,1]; \Real) : 0 \leq g \leq 1 \, \forall t \in [0,1]$,
\begin{equation*}
g(t)= \left\{\begin{array}{rcl}
	0 &\mbox{ if } t \in [0, \frac{1}{9}]  \,,\\
	1 &\mbox{ if } t \in [\frac{1}{2},1]  \,,\\
\end{array}\right.
\end{equation*}
and $g_n(x,y) \in C^\infty(\Real^2_+), n \geq 1$ 
\begin{equation}\label{cutoff}
g_n(x,y)= \left\{\begin{array}{rcl}
	&1 &\mbox{ if } \sqrt{x^2+y^2} \leq n  \,,\\
	&g(\frac{\ln\frac{n^2}{\sqrt{x^2+y^2}}}{\ln n}) &\mbox{ if } n < \sqrt{x^2 + y^2} < n^2 \,,\\
	&0 &\mbox{ if } \sqrt{x^2+y^2} \geq n^2  \,.\\
\end{array}\right.
\end{equation}
\begin{Lemma}\label{approximate func}
$g_n \xrightarrow[n \to \infty]{} 1$ pointwise and $\partial_x g_n \xrightarrow[n \to \infty]{} 0, \, \, \partial_x^2g_n \xrightarrow[n \to \infty]{} 0, \, \,  \partial_yg_n \xrightarrow[n \to \infty]{} 0$ in $L^2(\Real^2_+)$.
\end{Lemma}
\begin{proof}
The pointwise convergence is obvious due to the definition of $g_n$.

Compute  partial derivatives of $g_n$, we have
$$ \partial_x g_n(x,y) = \left\{\begin{array}{rcl}
	&0 &\mbox{ if } \sqrt{x^2+ y^2} \in [0, n] \cup [ n^2, \infty)  \,,\\
	&-\frac{1}{\ln n} \, \frac{x}{x^2+y^2} \, g'(\frac{\ln\frac{n^2}{\sqrt{x^2+y^2}}}{\ln n}) &\mbox{ if } \sqrt{x^2+y^2} \in (n,n^2)  \,,\\
\end{array}\right.
$$
$$ \partial_y g_n(x,y) = \left\{\begin{array}{rcl}
	&0 &\mbox{ if } \sqrt{x^2+ y^2} \in [0, n] \cup [ n^2, \infty)  \,,\\
	&-\frac{1}{\ln n} \, \frac{y}{x^2+y^2} \, g'(\frac{\ln\frac{n^2}{\sqrt{x^2+y^2}}}{\ln n}) &\mbox{ if } \sqrt{x^2+y^2} \in (n,n^2)  \,,\\
\end{array}\right.
$$
$$ \partial_x^2 g_n(x,y) = \left\{\begin{array}{rcl}
	&0 &\mbox{ if } \sqrt{x^2+ y^2} \in [0, n] \cup [ n^2, \infty)  \,,\\
	& f(x,y) &\mbox{ if } \sqrt{x^2+y^2} \in (n,n^2)  \,,\\
\end{array}\right.
$$
where $$ f(x,y) = \frac{1}{\ln^2 n} \, \frac{x^2}{(x^2+y^2)^2} \, g''(\frac{\ln\frac{n^2}{\sqrt{x^2+y^2}}}{\ln n}) + \frac{1}{\ln n} \, \frac{x^2 - y^2}{(x^2+y^2)^2} \, g'(\frac{\ln\frac{n^2}{\sqrt{x^2+y^2}}}{\ln n}) \,.$$
Passing to polar coordinates, we estimate
\begin{equation*}
\begin{aligned}
 \int_{\Real^2_+} |\partial_x g_n|^2 \der x \der y &= \int_0^\pi \int_n^{n^2} \frac{\cos^2\varphi}{r^2\ln^2n } |g'(\frac{\ln n^2 - \ln r}{\ln n})|^2 r \der r \der\varphi \,,\\
 &= \frac{\pi}{2} \frac{1}{\ln n} \int_0^1 |g'(t)|^2 \der t \xrightarrow[n \to \infty]{} 0 \,,
\end{aligned}
\end{equation*}
\begin{equation*}
\begin{aligned}
 \int_{\Real^2_+} |\partial_y g_n|^2 \der x \der y &= \int_0^\pi \int_n^{n^2} \frac{\sin^2\varphi}{r^2\ln^2n } |g'(\frac{\ln n^2 - \ln r}{\ln n})|^2 r \der r \der\varphi \,,\\
 &= \frac{\pi}{2} \frac{1}{\ln n} \int_0^1 |g'(t)|^2 \der t \xrightarrow[n \to \infty]{} 0 \,.
\end{aligned}
\end{equation*}
\begin{equation*}
\begin{aligned}
 \int_{\Real^2_+} |\partial_x^2 g_n|^2 \der x \der y &\leq 2 \int_0^\pi \int_n^{n^2} \frac{\cos^4\varphi}{r^4\ln^4n } |g''(\frac{\ln n^2 - \ln r}{\ln n})|^2 r \der r \der\varphi \,\\
 &\quad + 2 \int_0^\pi \int_n^{n^2} \frac{\cos^22\varphi}{r^4\ln^2n }|g'(\frac{\ln n^2 - \ln r}{\ln n})|^2 r \der r \der\varphi \,,\\
 &= \frac{3\pi}{4} \frac{1}{n^2\ln^3 n} \int_0^1 |g''(t)|^2 \der t  + \frac{\pi}{n^2 \ln n} \int_0^1 |g'(t)|^2 \der t \xrightarrow[n \to \infty]{} 0 \,.
\end{aligned}
\end{equation*}
Hence, the proof is completed.

\end{proof}

Putting
$$ A_\epsilon := \int_{\Real^2_+} (\epsilon^2 w_{11}^2 + \epsilon^2 w_{12}^2 + 4\delta\epsilon \Re w_{12} + \epsilon^2w_{21}^2 + \epsilon^2 w_{22}) \, \der x \der y  \,,$$
we consider the existence of the discrete spectrum of $H_\epsilon$ depending on the sign of $A_\epsilon$.
\begin{Theorem}\label{stability}
Let  $w_{ij} \in L^1(\Real^2_+) \cap L^\infty(\Real^2_+), w_{ij}= \overline{w_{ji}} $ for  $i,j = 1,2$ and vanishing at infinity. If
$A_\epsilon < 0$ then the discrete spectrum  of $H_\epsilon$  is non-empty.
\end{Theorem}
\begin{proof}
Denote $$ \psi_n := \begin{psmallmatrix} g_n(x,y) \\ g_n(x,y) \end{psmallmatrix} \,,$$ 
here $g_n$ is previously defined in \eqref{cutoff}. Thus, it deduces that $\psi_n \in \Dom(H_\epsilon)$, we have
\begin{equation}
\begin{aligned}
\|H\psi_n\|^2 - \delta^2 \|\psi_n\|^2 &= \|(-i \partial_y  + \epsilon w_{11}) g_n + (-\partial^2_x + \delta + \epsilon w_{12}) g_{n}\|^2 \\
&\quad  + \|(-\partial^2_x + \delta + \epsilon w_{21})g_{n} + (i \partial_y + \epsilon w_{22}) g_n\|^2  - 2 \delta^2\|g_n\|^2\,,\\
&\xrightarrow[n \to \infty]{} \int_{\Real^2_+} (\epsilon^2 w_{11}^2 + \epsilon^2 w_{12}^2 + 4\delta\epsilon \Re w_{12} + \epsilon^2w_{21}^2 + \epsilon^2 w_{22})  \, \der x \der y  \,,
\end{aligned}
\end{equation}
due to the dominated convergence theorem and Lemma \ref{approximate func}.

If $A_\epsilon < 0$ then there exist $n* \in \mathbb{N}, B_\epsilon < 0$ such that for all $ n \geq n*$ : 
$\|H\psi_n\|^2 - \delta^2 \|\psi_n\|^2 < B_\epsilon < 0$. As a result, the discrete spectrum of $H_\epsilon$ is non-empty or there exists at least one eigenvalue  in $(-\delta, \delta)$ of $H_\epsilon$. It concludes the proof.

\end{proof}
Theorem \ref{stability} demonstrates the spectral instability for a sequence of operators $H_\epsilon$ when there always exists  an electromagnetic potential $W$ such that $\sigma(H_\epsilon) \neq \sigma(H) = \sigma(H_{\epsilon=0})$.

Another comment referring to the case $ w_{11}=w_{22}=0, \Real\ni w_{12}=w_{21} \geq 0 $ is if there exists the discrete spectrum of $H_\epsilon$ then that has to be embedded eigenvalues lying in the essential spectrum  of $H_\epsilon$. Consequently, the discrete spectrum of the operator is empty in this case.
This result is  derived by using the analogous arguments given in \cite{KA}.
\begin{Corollary}\label{discrete2}
If $\int_{\Real^2_+}\Re w_{12} \ \der x \der y < 0$ then there exists $\epsilon(\delta) > 0$ such that for all $\epsilon \in (0, \epsilon(\delta))$ the discrete spectrum of $H_\epsilon$ is non-empty.
\end{Corollary}
\begin{proof}
The proof is directly deduced from Theorem \ref{stability}. Here we can choose 
$$ \epsilon(\delta) := \frac{-4 \delta \int_{\Real^2_+}\Re w_{12}}{\|w_{11}\|^2 + \|w_{22}\|^2 + \|w_{12}\|^2 + \|w_{21}\|^2 } \, .$$
It is apparent that   \quad \\
$$\lim_{\delta\rightarrow 0} \epsilon(\delta) =0, \quad \quad \quad\quad \quad \lim_{\delta\rightarrow \infty} \epsilon(\delta) = \infty \, .$$
The proof is completed.

\end{proof}
It is noteworthy that the sufficient conditions posed on $\Re w_{12}$ in Corollary \ref{discrete2} and in Theorem \ref{discrete1} are not identical. The former relies on the integral of $\Re w_{12}$ in $\Real^2_+$, while the latter requires  rigorous values of $V = \Re w_{12}$ in a certain range $[a,b]^2 \subset\subset \Real^2_+$.

\begin{Remark}
The spectral properties of the semi-Dirac semi-Laplacian are unstable and mismatched if we alter the half space $\Real^2_+$ into the other one $(0, \infty)\times\Real$ due to the fact that the corresponding semi-Dirac operator with Dirichlet boundary conditions  is also closed, symmetric but not self-adjoint in the half space $(0, \infty)\times\Real$. It can be seen as a distinction between that of the semi-Dirac semi-metals and  Dirac operators under some specific sorts of boundary conditions.
\end{Remark}

\subsection*{Acknowledgment}
The author is grateful to D. Krej\v{c}i\v{r}\'ik for useful comments and was partially supported by the EXPRO grant No.~20-17749X
of the Czech Science Foundation.

%
%

\providecommand{\bysame}{\leavevmode\hbox to3em{\hrulefill}\thinspace}
\providecommand{\MR}{\relax\ifhmode\unskip\space\fi MR }
\providecommand{\MRhref}[2]{%
  \href{http://www.ams.org/mathscinet-getitem?mr=#1}{#2}
}
\providecommand{\href}[2]{#2}

\end{document}